\documentclass[a4paper]{article}

\usepackage{amsmath,amssymb,amsthm,graphicx,fixmath,tikz,latexsym,color,xspace}
\usepackage[margin=1in]{geometry}

\usepackage{graphicx}
\usepackage{amsmath}
\usepackage{amssymb}
\usepackage{xspace}
\usepackage{microtype}

\usepackage{amsthm}
\theoremstyle{plain}
\newtheorem{theorem}{Theorem}
\newtheorem{lemma}{Lemma}
\theoremstyle{definition}

\graphicspath{{./images/}}

\newcommand{\etal}{\emph{et al.}\xspace}

\newcommand{\graph}{half-\ensuremath{\Theta_6}-graph\xspace}
\newcommand{\canon}[2]{\ensuremath{T_{#1 #2}}}
\newcommand{\Vis}{\mathord{\textrm{Vis}}}

\title{Competitive Local Routing with Constraints
  \thanks{Research supported by NSERC, the Ontario Ministry of Research and Innovation, Carleton University's President's 2010 Doctoral Fellowship, the Carleton-Fields
Postdoctoral Award, the Danish Council for Independent Research, Natural Sciences, grant DFF-1323-00247, and JST ERATO Grant Number JPMJER1305, Japan.}
  \thanks{An extended abstract containing some of the results in this paper appeared in the 26th International Symposium on Algorithms and Computation (ISAAC 2015)~\cite{BFRV2015Routing}.}}

\author{%
  Prosenjit~Bose,%
    \thanks{School of Computer Science, Carleton University, Ottawa, Canada, 
    \texttt{jit@scs.carleton.ca}, \texttt{sander@cg.scs.carleton.ca}}\, 
  Rolf~Fagerberg,%
    \thanks{Department of Mathematics and Computer Science, University of Southern Denmark, Odense, Denmark, 
    \texttt{rolf@imada.sdu.dk}}\, 
  Andr\'e~van~Renssen,%
    \thanks{National Institute of Informatics, Tokyo, Japan, 
    \texttt{andre@nii.ac.jp}}\, 
    \thanks{JST, ERATO, Kawarabayashi Large Graph Project}\, 
    and Sander~Verdonschot\footnotemark[3]%
}

\date{}

\begin{document}

\maketitle

\begin{abstract}
Let $P$ be a set of $n$ vertices in the plane and $S$ a set of non-crossing line segments between vertices in $P$, called constraints. Two vertices are visible if the straight line segment connecting them does not properly intersect any constraints. The constrained $\Theta_m$-graph is constructed by partitioning the plane around each vertex into $m$ disjoint cones, each with aperture $\theta = 2 \pi/m$, and adding an edge to the `closest' visible vertex in each cone. We consider how to route on the constrained $\Theta_6$-graph. We first show that no deterministic 1-local routing algorithm is $o(\sqrt{n})$-competitive on all pairs of vertices of the constrained $\Theta_6$-graph. After that, we show how to route between any two visible vertices of the constrained $\Theta_6$-graph using only 1-local information. Our routing algorithm guarantees that the returned path is 2-competitive. 
Additionally, we provide a 1-local 18-competitive routing algorithm for visible vertices in the constrained half-$\Theta_6$-graph, a subgraph of the constrained $\Theta_6$-graph that is equivalent to the Delaunay graph where the empty region is an equilateral triangle. To the best of our knowledge, these are the first local routing algorithms in the constrained setting with guarantees on the length of the returned path.
\end{abstract}

\section{Introduction}
A fundamental problem in any graph is the question of how to route a message from one vertex to another. What makes this more challenging is that often in a network the routing strategy must be \emph{local}. Informally, a routing strategy is \emph{local} when the routing algorithm must decide which vertex to forward a message to based solely on knowledge of the source and destination vertex, the current vertex and all vertices directly connected to the current vertex. Routing algorithms are considered \emph{geometric} when the graph that is routed on is embedded in the plane, with edges being straight line segments connecting pairs of vertices and weighted by the Euclidean distance between their endpoints. Geometric routing algorithms are important in wireless sensor networks \mbox{(see \cite{G09} and \cite{R09}} for surveys of the area) since they offer routing strategies that use the coordinates of the vertices to guide the search, instead of the more traditional routing tables.  

Most of the research has focused on the situation where the network is constructed by taking a subgraph of the complete Euclidean graph, i.e. the graph that contains an edge between every pair of vertices and the length of this edge is the Euclidean distance between the two vertices. We study this problem in a more general setting with the introduction of line segment {\em constraints}. Specifically, let $P$ be a set of vertices in the plane and let $S$ be a set of line segments between vertices in $P$, with no two line segments intersecting properly. The line segments of $S$ are called \emph{constraints}. Two vertices $u$ and $v$ can \textit{see each other} if and only if either the line segment $uv$ does not properly intersect any constraint or $u v$ is itself a constraint. If two vertices $u$ and $v$ can see each other, the line segment $uv$ is a \emph{visibility edge}. The \emph{visibility graph} of $P$ with respect to a set of constraints $S$, denoted $\Vis(P,S)$, has $P$ as vertex set and all visibility edges as edge set. In other words, it is the complete graph on $P$ minus all non-constraint edges that properly intersect one or more constraints in~$S$.

This natural extension allows for more realistic network modeling by excluding edges that cannot be used, such as ones crossing mountain ranges or areas of high interference which would scramble the message if used. As such, this setting has been studied extensively within the context of motion planning amid obstacles. Clarkson \cite{C87} was one of the first who studied this problem and showed how to construct a $(1+\epsilon)$-spanner of $\Vis(P,S)$ with a linear number of edges. A subgraph $H$ of $G$ is called a $t$-spanner of $G$ (for $t\geq 1$) if for each pair of vertices $u$ and $v$, the shortest path in $H$ between $u$ and $v$ has length at most $t$ times the shortest path in $G$ between $u$ and $v$. The smallest value $t$ for which $H$ is a $t$-spanner is the \emph{spanning ratio} or \emph{stretch factor} of $H$. Following Clarkson's result, Das \cite{D97} showed how to construct a spanner of $\Vis(P,S)$ with constant spanning ratio and constant degree. Bose and Keil \cite{BK06} showed that the Constrained Delaunay Triangulation (which contains an edge between two visible vertices $u$ and $v$ if and only if $u v$ is a constraint or there exists a circle with $u$ and $v$ on its boundary that contains no vertices visible to $u$ and $v$ in its interior) is a 2.42-spanner of $\Vis(P,S)$. Recently, the constrained \graph (which is identical to the constrained Delaunay graph whose empty visible region is an equilateral triangle, a formal definition follows in Section~\ref{sec:Preliminaries}) was shown to be a plane 2-spanner of $\Vis(P,S)$~\cite{BFRV12Constrained} and all constrained $\Theta$-graphs with at least 6 cones were shown to be spanners as well~\cite{BR14}. 

However, though it is known that these graphs contain short paths, it is not known how to route in a local fashion. In other words, other than by running some global shortest path algorithm or flooding the network with messages, the vertices are still unable to communicate with each other. To address this issue, we look at $k$-local routing algorithms in the constrained setting, i.e. routing algorithms that must decide which vertex to forward a message to based solely on knowledge of the source and destination vertex, the current vertex and all vertices that can be reached from the current vertex by following at most $k$ edges. Furthermore, we require our algorithms to be \emph{competitive}, i.e. the length of the returned path needs to be related to the length of the shortest path in the graph. 

In the unconstrained setting, there exists a 1-local 0-memory routing algorithm that is 2-competitive on the $\Theta_6$-graph and $5/\sqrt{3}$-competitive on the \graph (the $\Theta_6$-graph consists of the union of two \graph{s})~\cite{BFRV2015RoutingJournal}. In the same paper, the authors also show that these ratios are the best possible, i.e. there are matching lower bounds. 

In this paper, we show that the situation in the constrained setting is quite different: no deterministic 1-local routing algorithm is $o(\sqrt{n})$-competitive on all pairs of vertices of the constrained $\Theta_6$-graph, regardless of the amount of memory (defined in Section~\ref{sec:Preliminaries}) it is allowed to use. This shows that routing in the constrained setting is considerably harder than in the unconstrained setting. 

Despite this lower bound, we describe a 1-local 0-memory routing algorithm between any two \emph{visible} vertices of the constrained $\Theta_6$-graph that guarantees that the length of the path traveled is at most 2 times the Euclidean distance between the source and destination. Additionally, we provide a 1-local $O(1)$-memory 18-competitive routing algorithm between any two visible vertices in the constrained \graph. To the best of our knowledge, these are the first local routing algorithms in the constrained setting with guarantees on the path length.

\section{Preliminaries}
\label{sec:Preliminaries}
We define a \emph{cone} $C$ to be the region in the plane between two rays originating from a single vertex. This vertex is referred to as the apex of the cone. We let six rays originate from each vertex, with angles to the positive $x$-axis being multiples of $\pi / 3$ (see Figure~\ref{fig:Cones}). Each pair of consecutive rays defines a cone. We write $C_i^u$ to indicate the $i$-th cone of a vertex $u$, or $C_i$ if the apex is clear from the context. For ease of exposition, we only consider point sets in general position: no two vertices define a line parallel to one of the rays that define the cones and no three vertices are collinear. 

\begin{figure}[ht]
  \begin{minipage}[t]{0.46\linewidth}
    \begin{center}
      \includegraphics{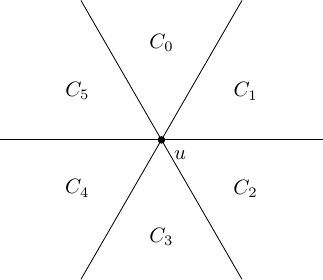}
    \end{center}
    \caption{The cones having apex $u$ in the $\Theta_6$-graph.}
    \label{fig:Cones}
  \end{minipage}
  \hspace{0.03\linewidth}
  \begin{minipage}[t]{0.46\linewidth}
    \begin{center}
      \includegraphics{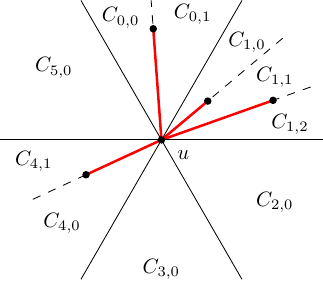}
    \end{center}
    \caption{The subcones having apex $u$ in the constrained $\Theta_6$-graph. Constraints are shown as thick red line segments.}
    \label{fig:ConstrainedCones}
  \end{minipage}
\end{figure}

Let vertex $u$ be an endpoint of a constraint and let the other endpoint lie in cone $C_i^u$. The lines through all such constraints split $C_i^u$ into several \emph{subcones} (see Figure~\ref{fig:ConstrainedCones}). We use $C_{i, j}^u$ to denote the $j$-th subcone, in clockwise order, of $C_i^u$. When a constraint $c = (u, v)$ splits a cone of $u$ into two subcones, we define $v$ to lie in both of these subcones. We consider a cone that is not split to be a single subcone.

The constrained $\Theta_6$-graph is constructed as follows: for each subcone $C_{i, j}$ of each vertex~$u$, add an edge from~$u$ to the closest visible vertex in that subcone, where distance is measured along the bisector of the original cone, not the subcone (see Figure~\ref{fig:Projection}). More formally, we add an edge between two vertices $u$ and $v$ if $v$ can see $u$, $v \in C_{i, j}$, and for all vertices $w \in C_{i, j}$ that can see $u$, $|u v'| \leq |u w'|$, where $v'$ and $w'$ denote the orthogonal projection of $v$ and $w$ on the bisector of $C_i$. Note that our general position assumptions imply that each vertex adds at most one edge per subcone to the graph. 

\begin{figure}[ht]
  \begin{center}
    \includegraphics{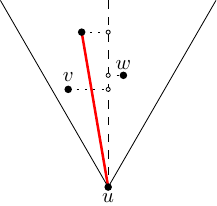}
  \end{center}
  \caption{Three vertices are projected onto the bisector of a cone of $u$. Vertex $v$ is the closest vertex in the left subcone and $w$ is the closest vertex in the right subcone.} 
  \label{fig:Projection}
\end{figure}

Next, we define the constrained \graph. This is a generalized version of the \graph as described by Bonichon \etal\cite{BGHI10}. The constrained \graph is similar to the constrained $\Theta_6$-graph with one major difference: edges are only added in every second cone. More formally, its cones are categorized as positive and negative. Let $(C_0, \overline{C}_2, C_1, \overline{C}_0, C_2, \overline{C}_1)$ be the sequence of cones in counterclockwise order starting from the positive $y$-axis (see Figure~\ref{fig:ConesHalfGraph}). The cones $C_0$, $C_1$, and $C_2$ are called \emph{positive} cones and $\overline{C}_0$, $\overline{C}_1$, and $\overline{C}_2$ are called \emph{negative} cones. We add edges only in the positive cones (and their subcones). Note that by using addition and subtraction modulo 3 on the indices, the positive cone $C_i$ has negative cone $\overline{C}_{i+1}$ as clockwise next cone and negative cone $\overline{C}_{i-1}$ as counterclockwise next cone. A similar statement holds for negative cones. We use $C^u_i$ and $\overline{C}^u_i$ to denote cones $C_i$ and $\overline{C}_i$ with apex $u$. For any two vertices $u$ and $v$, we have $v \in C^u_i$ if and only if $u \in \overline{C}^v_i$ (see Figure~\ref{fig:ConesHalfGraph}). Analogous to the subcones defined for the $\Theta_6$-graph, constraints can split cones into subcones. We call a subcone of a positive cone a positive subcone and a subcone of a negative cone a negative subcone (see Figure~\ref{fig:ConstrainedConesHalfGraph}). We look at the undirected version of these graphs, i.e. when an edge is added, both vertices are allowed to use it. This is consistent with previous work on $\Theta$-graphs. 

\begin{figure}[h!]
  \begin{minipage}[t]{0.46\linewidth}
    \begin{center}
      \includegraphics{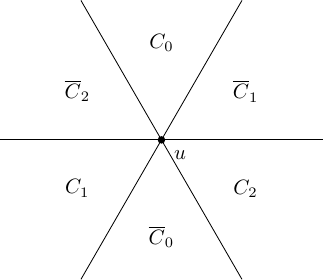}
    \end{center}
    \caption{The cones having apex $u$ in the \graph.}
    \label{fig:ConesHalfGraph}
  \end{minipage}
  \hspace{0.05\linewidth}
  \begin{minipage}[t]{0.46\linewidth}
    \begin{center}
      \includegraphics{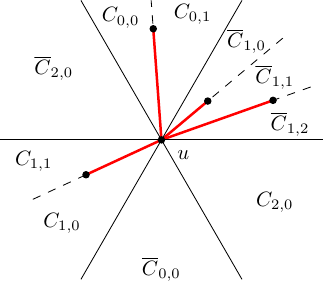}
    \end{center}
    \caption{The subcones having apex $u$ in the constrained \graph. Constraints are shown as thick red line segments.}
    \label{fig:ConstrainedConesHalfGraph}
  \end{minipage}
\end{figure}

Given a vertex $w$ in a positive cone $C^u_i$ of vertex $u$, we define the \emph{canonical triangle} \canon{u}{w} to be the triangle defined by the borders of $C^u_i$ (not the borders of the subcone of $u$ that contains $w$) and the line through $w$ perpendicular to the bisector of $C^u_i$ (see Fig.~\ref{fig:CanonicalTriangle}). Note that for each pair of vertices there exists a unique canonical triangle. 

\begin{figure}[ht]
  \begin{center}
    \includegraphics{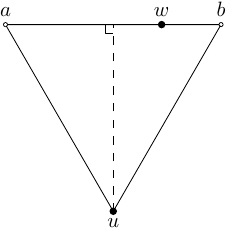}
  \end{center}
  \caption{The canonical triangle \canon{u}{w}.}
  \label{fig:CanonicalTriangle}
\end{figure}

Next, we define our routing model. A deterministic routing algorithm is $k$-local and uses $m$-memory, if the vertex to which a message is forwarded from the current vertex $u$ is a function of $s$, $t$, $N_k(u)$, and $M$, where $s$ and $t$ are the source and destination vertex, $N_k(u)$ is the $k$-neighborhood of $u$ and $M$ is a memory of size $m$, stored with the message. The $k$-neighborhood of a vertex $u$ is the set of vertices in the graph that can be reached from $u$ by following at most $k$ edges. For our purposes, we consider a unit of memory to consist of $\log_2 n$ bits or a point in $\mathbb{R}^2$. Our model also assumes that the only information stored at each vertex of the graph is $N_k(u)$. Since our graphs are geometric, we identify each vertex by its coordinates in the plane. Unless otherwise noted, all routing algorithms we consider in this paper are deterministic $0$-memory algorithms.

There are essentially two notions of \emph{competitiveness} of a routing algorithm on a subgraph of the visibility graph. One is to look at the Euclidean shortest path between the two vertices, i.e. the shortest path in the visibility graph, and the other is to compare the routing path to the shortest path in the subgraph. A routing algorithm is {\em $c$-competitive with respect to the Euclidean shortest path (resp. shortest path in the subgraph)} provided that the total distance traveled by the message is not more than $c$ times the Euclidean shortest path length (resp. shortest path length) between source and destination. The \emph{routing ratio} of an algorithm is the smallest $c$ for which it is $c$-competitive. 

Since the shortest path in the subgraph between two vertices is at least as long as the Euclidean shortest path between them, an algorithm that is $c$-competitive with respect to the Euclidean shortest path is also $c$-competitive with respect to the shortest path in the subgraph. We use competitiveness with respect to the Euclidean shortest path when proving upper bounds and with respect to the shortest path in the subgraph when proving lower bounds.

Furthermore, we want to be able to talk about points at intersections of lines, thus we distinguish between \emph{vertices} and \emph{points}. A \emph{point} is any point in $\mathbb{R}^2$, while a \emph{vertex} is part of the input.

\section{Lower Bound on Local Routing}
We modify the proof by Bose~\etal~\cite{BBCDFLMM00} (that shows that no deterministic routing algorithm is $o(\sqrt{n})$-competitive for all triangulations) to show the following lower bound. 

\begin{theorem}
\label{theo:LowerBound}
  No deterministic 1-local routing algorithm is $o(\sqrt{n})$-competitive with respect to the shortest path on all pairs of vertices of the $\Theta_6$-graph of size $n$, regardless of the amount of memory it is allowed to use. 
\end{theorem}
\begin{proof}[Proof.]
  The following construction is illustrated in Figure~\ref{fig:RoutingLowerBoundConstruction}a-e. Consider a $c \times c$ grid of vertices for an integer $c$ and shift every second row to the right by half a unit. We stretch the grid, such that each horizontal line segment has length $2c$. Next, we replace each horizontal line segment by a constraint to prevent vertical visibility edges and we remove all other line segments. After that, we add two additional vertices, source $s$ and destination $t$, centered horizontally at one unit below the bottom row and one unit above the top row, respectively. 
  
  \begin{figure}[ht]
    \begin{center}
      \includegraphics{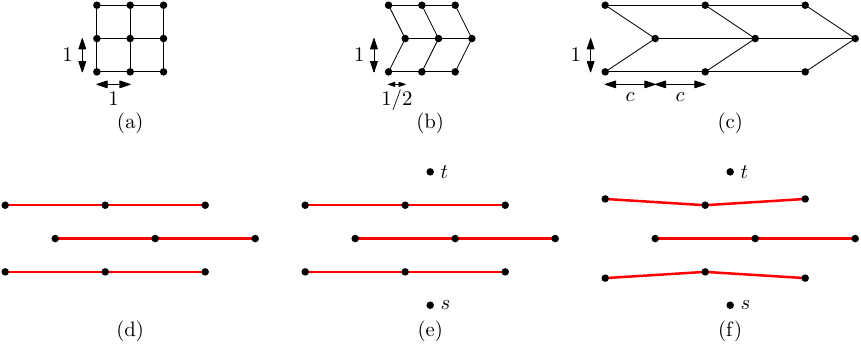}
    \end{center}
    \caption{Constructing the lower bound: (a) the gird, (b) after shifting, (c) after stretching (d) adding the constraints, (e) adding $s$ and $t$, (f) conforming to general position.}
    \label{fig:RoutingLowerBoundConstruction}
  \end{figure}

  To conform to our general position assumption, we move all vertices by at most some arbitrarily small amount $\epsilon$, such that no two vertices define a line parallel to one of the rays that define the cones and no three vertices are collinear (see Figure~\ref{fig:RoutingLowerBoundConstruction}f). As part of this move, we ensure that each vertex on the bottom row has $s$ as its closest vertex in cone $C_2$ or $C_4$ (depending on whether it lies to the right or left of $s$), and that each vertex on the top row has $t$ as its closest vertex in cone $C_1$ or $C_5$ (again depending on whether it lies to the left or right of $t$). This can be done e.g.\ by placing the bottom row on the upper hull of an ellipse and placing the top row on the lower hull of an ellipse. On this point set and these constraints, we build the constrained $\Theta_6$-graph $G$ (see Figure~\ref{fig:RoutingLowerBoundGrid}). Note that vertical edges only appear at the left and right grid boundaries.

  \begin{figure}[ht]
    \begin{center}
      \includegraphics{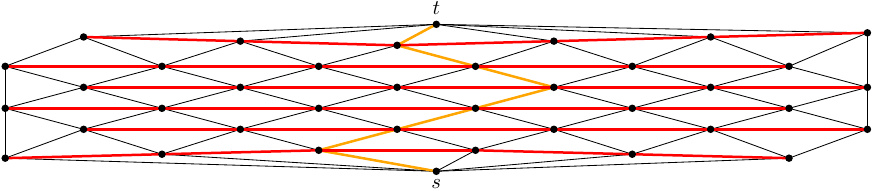}
    \end{center}
    \caption{The constrained $\Theta_6$-graph starting from a grid, using horizontal constraints to block vertical edges, and the orange path of the routing algorithm.}
    \label{fig:RoutingLowerBoundGrid}
  \end{figure}

  Consider any deterministic 1-local $\infty$-memory routing algorithm and let $\pi$ be the path this algorithm takes when routing from $s$ to $t$. We note that by construction, $\pi$ consists of at least $c + 1$ steps. If $\pi$ consists of more than $c \sqrt{c}$ non-vertical steps, we truncate it after the first $c \sqrt{c}$ non-vertical steps.
Thus, in the remainder of this proof, we consider only paths having at most $k$ non-vertical steps for $k \leq c \sqrt{c}$. The overall idea of the proof is to reduce $G$ to a $\Theta_6$-graph $G'$ of size $\Theta(c + k)$ in a way which does not change the path $\pi$ (up to its truncation point, if present) taken by the algorithm, and then to show that $\pi$ is not $o( \sqrt{c + k})$-competitive with respect to the shortest path in~$G'$. This proves that no deterministic 1-local $\infty$-memory routing algorithm can be $o( \sqrt{n})$-competitive with respect to the shortest path on all $\Theta_6$-graphs.

To construct~$G'$, we define the \emph{surroundings} of a vertex~$v$ on $\pi$ to be $v$ itself, the vertices connected to it by either an edge or a constraint in~$G$, and the constraints in~$G$ between these vertices. Thus, for $v$ in the interior of~$G$, its surroundings are hexagonal in shape and contain seven vertices and four constraints (see Figure~\ref{fig:RoutingLowerBoundGrid}). Informally, the union of the surroundings of vertices of~$\pi$ can be seen as sweeping this hexagonal shape along~$\pi$. For $v$ on the border of~$G$, its surroundings are slightly smaller. For $s$ and $t$, their surroundings constitute the bottom and top row, including the constraints in these rows. We let~$G'$ be the $\Theta_6$-graph constructed on the union of the surroundings of all vertices of~$\pi \cup \{ t \}$ (the inclusion of~$t$ is only relevant if $\pi$ was truncated). This construction is illustrated in Figure~\ref{fig:RoutingLowerBoundGraph}. Clearly, the graph~$G'$ has $O(c + k)$ vertices and constraints. It is easy to check that the 1-neighborhood of any vertex~$v$ on $\pi$ is the same in~$G'$ as in $G$, hence the routing algorithm must follow~$\pi$ also in $G'$.

  \begin{figure}[ht]
    \begin{center}
      \includegraphics{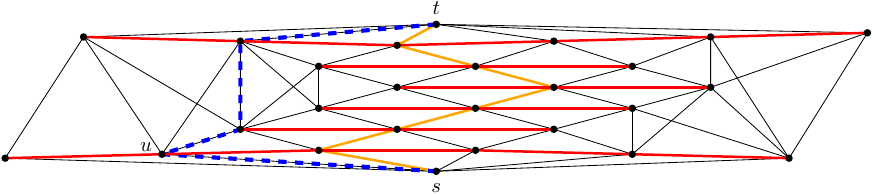}
    \end{center}
    \caption{The constrained $\Theta_6$-graph that looks the same from the orange path of the routing algorithm, but has an mostly vertical dashed blue path.}
    \label{fig:RoutingLowerBoundGraph}
  \end{figure}

The bottom row contains~$c$ vertices. We now consider the $2\sqrt{k}$ horizontally most central of these, that is, the first $\sqrt{k}$ vertices to the left of~$s$ and the first $\sqrt{k}$ vertices to the right of~$s$. Setting $c \geq 16$, the bottom row does contain at least these $2 \sqrt{k}$ vertices, by $k \leq c \sqrt{c}$. Setting $c$ a bit higher, we can assume that it contains $\Omega(1)$ more vertices at each end. Next, consider a vertical line through each of these $2 \sqrt{k}$ vertices. Let $\pi'$ be $\pi$ minus the vertices~$s$ and $t$. We say that a vertex of $\pi'$ \emph{touches} such a vertical line if its surroundings contain a point on that line. Hence, any vertex along $\pi'$ touches $O(1)$ vertical lines (see Figure~\ref{fig:RoutingLowerBoundGrid}). Since the vertical lines are $\Omega(1)$ grid positions away from the left and right sides of the grid, no vertical step of~$\pi'$ can touch any of these lines. Hence, the total number of line touches by the vertices along $\pi'$ is at most $O(k)$. Hence, on average, a line is touched $O(k/ \sqrt{k}) = O(\sqrt{k})$ times. This implies that there exists a vertical line that is touched $O(\sqrt{k})$ times. Let $u$ be vertex on the bottom row whose vertical line is touched the fewest number of times. 

We now prove that a `mostly vertical' path from $u$ to the top row is contained in $G'$, which will provide a path $G'$ between $s$ and $t$ much shorter than the path $\pi$ which the algorithm must follow. Assume first that the line of~$u$ is touched zero times. In the remainder of the proof, we set $c$ to be odd, such that vertices on the top and bottom row align horizontally. Since the minimal horizontal distance between vertices in the grid is $2c$, while the maximal vertical distance is $c$, $u$ can see exactly one vertex in $C_0$, namely the vertex it aligns horizontally with in the top row. Thus, there is a vertical edge between these two vertices in~$G'$. If the line of~$u$ is touched more than zero times, each touch covers a part of the line with some parts of the hexagonal shape. The coverings may overlap, and they give rise to a natural decomposition of the line into maximal covered segments with non-covered segments in between. A core observation is that a covered segment vertically extending $h$ grid levels can be traversed by $h$ zig-zag edges in~$G'$, of total length~$O(ch)$. Some examples of this are shown in Figure~\ref{fig:RoutingLowerBoundDetours}.

  \begin{figure}[h]
    \begin{center}
      \includegraphics{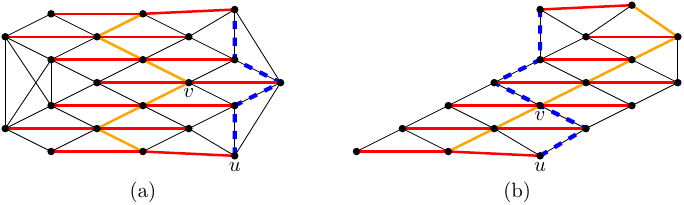}
    \end{center}
    \caption{Two examples of covered segments and their zig-zag detours: (a) when $\pi$ gets close but does not meet the vertical line through $u$, (b) when $\pi$ crosses the vertical line through $u$ once.}
    \label{fig:RoutingLowerBoundDetours}
  \end{figure}

Another core observation is that for each uncovered segment of the line, there will be a vertical edge in~$G'$ from the top vertex of the covered segment below to the bottom vertex of the covered segment above (again due to the vertex distances in the grid). Thus, the vertical edge from the case of zero touches is broken up by zig-zag shaped detours (one detour for each covered segment). The resulting path has length $O(c \sqrt{k})$, since the line through $u$ is touched by at most $O(\sqrt{k})$ vertices of $\pi$, each of which can cover only $O(1)$ grid levels of the line. Recalling that the edge from $s$ to $u$ has length at most $c \sqrt{k}$, we conclude that $G'$ contains a path from $s$ to $t$ of length $O(c \sqrt{k})$: Follow the edge from $s$ to $u$, follow the above path from $u$ to the top row of $G'$, and follow the edge to $t$.

  To complete the proof, we look at the number of non-vertical edges of $\pi$, i.e. $k$. If $k \leq c$, the routing path follows at least one vertical edge along the boundary of~$G$. It follows that $\pi$ has length at least $\Omega(c^2)$, as the left and right boundary of $G$ are at distance $\Omega(c^2)$ from $s$. Since the length of the mostly vertical path is $O(c \sqrt{k})$, $\pi$ is not $o(c/\sqrt{k})$-competitive on a graph of size $\Theta(c + k)$, which for $k \leq c$ implies that $\pi$ is not $o(\sqrt{c})$-competitive on a graph of size $\Theta(c)$. Hence, when we take $n = c$, the theorem is proven for this case. 
 
  If $k > c$, the length of $\pi$ is dominated by the non-vertical edges of length $c$, leading to a path length of $\Omega(c k)$. Since the length of the mostly vertical path is $O(c \sqrt{k})$, this implies that $\pi$ is not $o(\sqrt{k})$-competitive on a graph of size $\Theta(k)$. Hence, when we take $n = k$, the theorem is proven for this case. 

Thus, since $G'$ can be constructed for any deterministic 1-local routing algorithm, we have shown that no deterministic 1-local routing algorithm is $o(\sqrt{n})$-competitive on all pairs of vertices in a graph of size $O(n)$. 
\end{proof}

\section[Routing on the Constrained $\Theta_6$-Graph]{Routing on the Constrained $\boldsymbol{\Theta_6}$-Graph}
In this section, we provide a 1-local routing algorithm on the constrained $\Theta_6$-graph for any pair of visible vertices. Since the constrained $\Theta_6$-graph is the union of two constrained \graph{s}, we describe a routing algorithm for the constrained \graph for the case where the destination $t$ lies in a positive subcone of the source $s$. After describing this algorithm and proving that it is 2-competitive, we describe how to use it to route 1-locally on the constrained $\Theta_6$-graph. Throughout this section, we use the following auxiliary lemma proven by Bose~\etal~\cite{BFRV12Constrained}. We say that a region is \emph{empty} if it does not contain any vertices of $P$.

\begin{lemma}
  \label{lem:ConvexChain}
  Let $u$, $v$, and $w$ be three arbitrary points in the plane such that $u w$ and $v w$ are visibility edges and $w$ is not the endpoint of a constraint intersecting the interior of triangle $u v w$. Then there exists a convex chain of visibility edges from $u$ to $v$ in triangle $u v w$, such that the polygon defined by $u w$, $w v$ and the convex chain is empty and does not contain any constraints (see Fig~\ref{fig:VisiblePointInsideTriangle}).
\end{lemma}

\begin{figure}[ht]
  \begin{center}
    \includegraphics{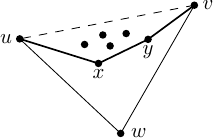}
  \end{center}
  \caption{A convex chain from $u$ to $v$ via $x$ and $y$.}
  \label{fig:VisiblePointInsideTriangle}
\end{figure}

Recall that when working on upper bounds, we use the notion of competitiveness with respect to the Euclidean shortest path: A routing algorithm is {\em $c$-competitive with respect to the Euclidean shortest path} provided that the total distance traveled by the message is not more than $c$ times the Euclidean shortest path length between source and destination. The \emph{routing ratio} of an algorithm \emph{with respect to the Euclidean shortest path} is the smallest $c$ for which it is $c$-competitive with respect to the Euclidean shortest path.

\subsection[Positive Routing on the Constrained \graph]{Positive Routing on the Constrained Half-$\boldsymbol{\Theta_6}$-Graph}
Before describing how to route on the constrained \graph when $t$ lies in a positive subcone of $s$, we first show that there exists a path in canonical triangle \canon{s}{t}. 

\begin{lemma}
  \label{lem:PathInTriangle}
  Given two vertices $u$ and $w$ such that $u$ and $w$ see each other and $w$ lies in a positive subcone $C^u_{i,j}$, there exists a path between $u$ and $w$ in the triangle \canon{u}{w} in the constrained \graph. 
\end{lemma}
\begin{proof}[Proof.]
  We assume without loss of generality that $w$ lies in $C^u_{0, j}$. We prove the lemma by induction on the area of the canonical triangle \canon{u}{w}. Formally, we perform induction on the rank of the triangle in the ordering, according to their area, of the canonical triangles \canon{x}{y} of all pairs of visible vertices $x$ and $y$. 

  \textbf{Base case:} If \canon{u}{w} is the smallest canonical triangle, then $w$ is the closest visible vertex to $u$ in a positive subcone of $u$. Hence there is an edge between $u$ and $w$ and this edge lies entirely inside \canon{u}{w}. 

  \textbf{Induction step:} We assume that the induction hypothesis holds for all pairs of vertices that can see each other and have a canonical triangle whose area is smaller than the area of \canon{u}{w}. If $u w$ is an edge in the constrained \graph, the induction hypothesis follows by the same argument as in the base case. If there is no edge between $u$ and $w$, let $v_0$ be the vertex closest to $u$ in the positive subcone that contains $w$, and let $a_0$ and $b_0$ be the upper left and right corner of \canon{u}{v_0} (see Figure~\ref{fig:ConvexChain}). We assume without loss of generality that $v_0$ lies to the left of $u w$. 

  \begin{figure}[ht]
    \begin{center}
      \includegraphics{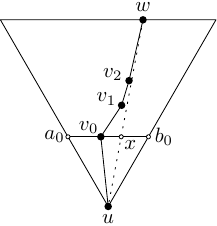}
    \end{center}
    \caption{An example of a convex chain from $v_0$ to $w$.}
    \label{fig:ConvexChain}
  \end{figure}

  Let $x$ be the intersection of $u w$ and $a_0 b_0$. By definition $x$ can see $u$ and $w$. Since $v_0$ is the closest visible vertex to $u$, $v_0$ can see $x$ as well. Otherwise Lemma~\ref{lem:ConvexChain} would give us a convex chain of vertices connecting $v_0$ to $x$, all of which would be closer and able to see $u$, contradicting that $v_0$ is the closest visible vertex to $u$. By applying Lemma~\ref{lem:ConvexChain} to triangle $v_0 x w$, a convex chain $v_0, v_1, ..., v_k = w$ of visibility edges connecting $v_0$ and $w$ exists and the region bounded by $x, v_0, v_1, ..., v_k = w$ is empty (see Figure~\ref{fig:ConvexChain}).

  Since every vertex $v_i$ is visible to vertex $v_{i+1}$, we can apply induction to each pair of consecutive vertices along the convex chain. Depending on whether $v_{i+1} \in C^{v_i}_0$ or $v_i \in C^{v_{i+1}}_1$, there exists a path between $v_i$ and $v_{i+1}$ in \canon{v_i}{v_{i+1}} or \canon{v_{i+1}}{v_i}. Since each of these triangles is contained in \canon{u}{w}, this gives us a path between $u$ and $w$ that lies inside \canon{u}{w}. 
\end{proof}

\noindent \textbf{Positive Routing Algorithm for the Constrained Half-$\boldsymbol{\Theta_6}$-Graph} \\
Next, we describe how to route from $s$ to $t$, when $s$ can see $t$ and $t$ lies in a positive subcone $C^s_{i,j}$ (see Figure~\ref{fig:PositiveRouting}): When we are at $s$, we follow the edge to the closest vertex in the subcone that contains $t$. When we are at any other vertex $u$, we look at all edges in the subcones of $C^u_i$ and all edges in the subcones of the adjacent negative cone $\overline{C}^u$ that is intersected by $s t$. An edge in a subcone of $\overline{C}^u$ is considered only if it does not cross $s t$. For example, in Figure~\ref{fig:PositiveRouting}, we do not consider the edge to $v_1$ since it lies in $\overline{C}^u$ and crosses $s t$. It follows that we can cross $s t$ only when we follow an edge in $C^u_i$. 

\begin{figure}[ht]
  \begin{center}
    \includegraphics{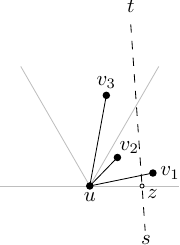}
  \end{center}
  \caption{An example of routing from $s$ to $t \in C^s_0$. The dashed line represents the visibility line between $s$ and $t$.}
  \label{fig:PositiveRouting}
\end{figure}

Let $z$ be the intersection of $s t$ and the boundary of $\overline{C}^u$ that is not a boundary of $C^u_i$. We follow the edge $u v$ that minimizes the unsigned angle $\angle z u v$. For example, in Figure~\ref{fig:PositiveRouting}, when we are at vertex $u$ we follow the edge to $v_2$ since, out of the two remaining edges $u v_2$ and $u v_3$, $\angle z u v_2$ is smaller than $\angle z u v_3$. We note that edges in $\overline{C}^u$ are added by the vertices in that cone, since $u$ lies in their positive cone $C$. We also note that during the routing process, $t$ does not necessarily lie in $C^u_i$. Finally, since the algorithm uses only information about the location of $s$ and $t$ and the neighbors of the current vertex, it is a 1-local routing algorithm. 

We proceed by proving that the above routing algorithm can always perform a step, i.e. at every vertex reached by the algorithm there exists an edge that is considered by the algorithm. 

\begin{lemma}
  \label{lem:RoutingStep}
  The routing algorithm can always perform a step in the constrained \graph.
\end{lemma}
\begin{proof}[Proof.]
  Given two vertices $s$ and $t$ such that $s$ and $t$ can see each other, we assume without loss of generality that $t \in C^s_0$. We maintain the following invariant (see Figure~\ref{fig:Invariant}): 

  \begin{quote}
    \textbf{Invariant} Let $x$ be the last intersection of an edge of the routing path with $s t$ (initially $x$ is $s$), let $v_0, ..., v_k$ denote the endpoints of the edges following $x$ as selected by the algorithm, and let $x'$ be the intersection of $s t$ and the horizontal line through $v_k$. The simple polygon defined by $x, v_0, ..., v_k, x'$ is empty and does not contain any constraints. 
  \end{quote}

  \begin{figure}[ht]
    \begin{center}
      \includegraphics{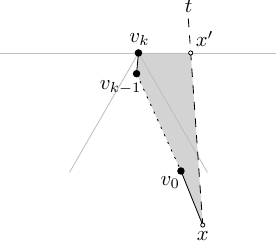}
    \end{center}
    \caption{By the invariant, the gray region is empty and does not contain any constraints.}
    \label{fig:Invariant}
  \end{figure}

  When the routing algorithm starts at $s$, it looks at the subcone that contains $t$. Since $t$ is visible from $s$, this subcone contains at least one visible vertex. Hence, it also contains a closest visible vertex $v_0$ and by construction, $s$ has an edge to $v_0$. Therefore, when the routing algorithm starts at $s$, it can follow an edge. 

  To see that the invariant is satisfied, we need to show that triangle $s v_0 x'$ is empty and does not contain any constraints in its interior. By construction $s$ cannot be the endpoint of any constraints in the interior of $s v_0 x'$, hence since $s x'$ and $s v_0$ are visibility edges, any constraint has at least one endpoint in $s v_0 x'$. Thus, it suffices to show that $s v_0 x'$ is empty. We prove this by contradiction, so assume that it is not empty. Since $s v_0$ and $s x'$ are visibility edges and by construction $s$ is not the endpoint of a constraint intersecting the interior of $s v_0 x'$, Lemma~\ref{lem:ConvexChain} gives us a convex chain of visibility edges between $v_0$ and $x'$. Since the region bounded by $s v_0$, $s x'$, and this chain is empty and does not contain any constraints, the vertex along this chain that is closest to $s$ is visible to $s$. However since every vertex in $s v_0 x'$ is closer to $s$ than $v_0$, this contradicts the fact that $v_0$ is the closest visible vertex to $s$. Hence, triangle $s v_0 x'$ must be empty and the invariant is satisfied. 

  When the routing algorithm is at vertex $u$ ($u \neq s$), we assume without loss of generality that $u$ lies to the left of $s t$. Let $h$ be the halfplane below the horizontal line through $t$ and let $h'$ be the halfplane to the left of $s t$. We need to show that $u$ has at least one edge in the union of $C^u_0 \cap h$ and $\overline{C}^u_1 \cap h \cap h'$. We first show that there exists a vertex that is visible to $u$ in the union of $C^u_0 \cap h$ and $\overline{C}^u_1 \cap h \cap h'$, by showing that such a vertex exists in the union of $C^u_0 \cap h \cap h'$ and $\overline{C}^u_1 \cap h \cap h'$. Since $t$ lies in this region, we know that it is not empty. Consider all vertices in this region and let $v$ be the vertex in this region that minimizes $\angle x' u v$. Note that we did not require there to be an edge between $u$ and $v$. Since $v$ minimizes $\angle x' u v$ and no constraint can cross $s t$ or $u x'$, $v$ is visible from $u$. We consider two cases: $v$ lies in a subcone of $C^u_0$ and $v$ lies in a subcone of $\overline{C}^u_1$. 

  If $v$ lies in $C^u_0 \cap h \cap h'$, it follows from Lemma~\ref{lem:PathInTriangle} and the fact that $v$ is visible from $u$ that there exists a path between $u$ and $v$ that lies inside \canon{u}{v}. Since \canon{u}{v} is contained in $C^u_0 \cap h$, there exists an edge in $C^u_0 \cap h$ and the routing algorithm can perform a step. 

  If $v$ lies in $\overline{C}^u_1 \cap h \cap h'$, it follows from Lemma~\ref{lem:PathInTriangle} and the fact that $v$ is visible from $u$ that there exists a path between $u$ and $v$ that lies inside \canon{v}{u}. Canonical triangle \canon{v}{u} intersects three cones of $u$ (see Figure~\ref{fig:RoutingStepExistence}): $C^u_0$, $\overline{C}^u_1$, and $C^u_2$. Since the routing algorithm follows edges in $C^u_0$ or $\overline{C}^u_1$, the routing path reaches $u$ by following edge $v_{k-1} u$ that lies in either $\overline{C}^u_0$ or $C^u_1$. This implies that $\canon{v}{u} \cap C^u_2$ is contained in the region of the invariant and is therefore empty. Hence, the first edge on the path from $u$ to $v$ lies in either $C^u_0 \cap h$ or $\overline{C}^u_1 \cap h \cap h'$ and the algorithm can perform a step. 

  \begin{figure}[ht]
    \begin{center}
      \includegraphics{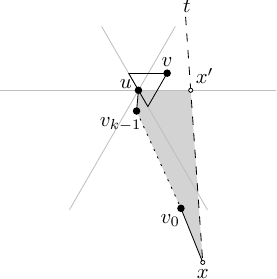}
    \end{center}
    \caption{By the invariant, the gray region is empty, so the path between $u$ and $v$ lies inside $\canon{v}{u} \cap (C^u_0 \cup \overline{C}^u_1)$.}
    \label{fig:RoutingStepExistence}
  \end{figure}

  It remains to show that after the algorithm takes a step, the invariant is satisfied at the new vertex $v$. Let $u v$ be the edge that the algorithm followed and let $x''$ be the intersection of $s t$ and the horizontal line through $v$. We consider three cases (see Figure~\ref{fig:MaintainingInvariant}): (a) $v$ lies in a subcone of $\overline{C}^u_1$, (b) $v$ lies in a subcone of $C^u_0$ and $u v$ does not cross $s t$, and (c) $v$ lies in a subcone of $C^u_0$ and $u v$ crosses $s t$. 

  \begin{figure}[ht]
    \begin{center}
      \includegraphics{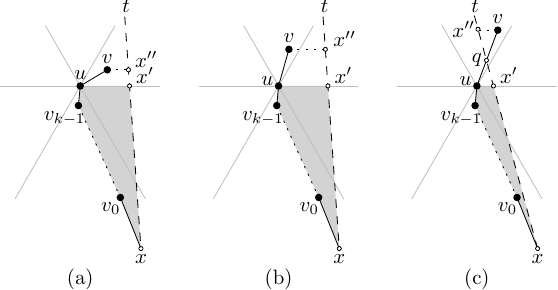}
    \end{center}
    \caption{The three types of steps the algorithm can take: (a) $v$ lies in a subcone of $\overline{C}^u_1$, (b) $v$ lies in a subcone of $C^u_0$ and $u v$ does not cross $s t$, and (c) $v$ lies in a subcone of $C^u_0$ and $u v$ crosses $s t$.}
    \label{fig:MaintainingInvariant}
  \end{figure}

  \textbf{Case (a):} If $v$ lies in a subcone of $\overline{C}^u_1$, we need to show that the quadrilateral $u v x'' x'$ is empty and does not contain any constraints (see Figure~\ref{fig:MaintainingInvariant}a). We first show that $u$ cannot be the endpoint of a constraint intersecting the interior of $u v x'' x'$. We prove this by contradiction, so assume it is and let $y$ be the other endpoint of the constraint. We first note that $\angle x' u y < \angle x' u v$. We look at $C^y_{1, j}$, the subcone of $C^y_1$ that lies below $u y$, and let $z$ be the lowest vertex in this subcone. If $u$ is the closest visible vertex in this subcone, $u y$ would be an edge, which contradicts that $v$ minimizes $\angle x' u v$. Otherwise, since $z$ is the lowest vertex in $C^y_{1, j}$, the visible region of \canon{z}{u} is empty and $u z$ is an edge. However, since $\angle x' u z < \angle x' u y < \angle x' u v$, we have a contradiction. Thus $u$ cannot be the endpoint of a constraint intersecting the interior of $u v x'' x'$. 

  Since $u$ is not the endpoint of a constraint intersecting the interior of $u v x'' x'$, and $u v$, $u x'$, and $x' x''$ are visibility edges, any constraint intersecting the interior of $u v x'' x'$ has at least one endpoint in $u v x'' x'$. Thus it suffices to show that $u v x'' x'$ is empty. We prove this by contradiction, so assume that $u v x'' x'$ is not empty and let $y$ be the lowest vertex in $u v x'' x'$. Let $C^y_{1, j}$ be the subcone of $C^y_1$ that contains $u$. Vertex $u$ is visible to $y$, since any constraint crossing $u y$ has an endpoint in $\overline{C}^u_1$ below $y$, contradicting that $y$ is the lowest vertex, or in the region bounded by $x, v_0, ..., v_{k-1}, u, x'$ which contradicts the invariant. Hence $y$ has an edge in $C^y_{1, j}$. This edge cannot be to $u$ since $\angle x' u y < \angle x' u v$. Since $y$ is the lowest vertex in $u v x'' x'$, it cannot have an edge to a vertex in $u v x'' x'$. Since by the invariant the region bounded by $x, v_0, ..., v_{k-1}, u, x'$ is empty, the edge of $y$ in $C^y_{1, j}$ must cross $u v$. However, this contradicts the fact that the constrained \graph is plane. Thus, $u v x'' x'$ is empty of both vertices and constraints. 

  \textbf{Case (b):} If $v$ lies in a subcone of $C^u_0$ and $u v$ does not cross $s t$, we again need to show that the quadrilateral $u v x'' x'$ is empty and does not contain any constraints (see Figure~\ref{fig:MaintainingInvariant}b). We first show that $u v x'' x'$ is empty. We prove this by contradiction, so assume that $u v x'' x'$ is not empty and let $y$ be the lowest vertex in $u v x'' x'$. We consider two cases: $y$ lies in $\overline{C}^u_1$ and $y$ lies in $C^u_0$. Since the case where $y$ lies in $\overline{C}^u_1$ is analogous to the Case~(a), we focus on the case where $y$ lies in a subcone of $C^u_0$.

  If $y$ lies in a subcone of $C^u_0$ and $y$ is visible to $u$, $u y$ would be an edge and $\angle x' u y < \angle x' u v$. So, assume that $y$ is not visible from $u$. This means that there is a constraint that crosses $u y$. Since the line $s t$ and the edges of the region bounded by $x, v_0, ..., v_{k-1}, u, x'$ are visibility edges, the lower endpoint of this constraint must lie in $x, v_0, ..., v_{k-1}, u, v, x''$. By the invariant, it cannot lie in $x, v_0, ..., v_{k-1}, u, x'$, so it must lie in $u v x'' x'$ and below $y$. However, this contradicts that $y$ is the lowest vertex in $u v x'' x'$. Since we arrived at a contradiction in both cases, we conclude that quadrilateral $u v x'' x'$ is empty. 

  Next, we show that $u v x'' x'$ does not contain any constraints. Since $u v x'' x'$ is empty, a the only way a constraint can intersect it, is when $u$ is one of its endpoints. Hence, it remains to show that $u$ cannot be the endpoint of a constraint intersecting the interior of $u v x'' x'$. We prove this by contradiction, so assume it is and let $y$ be the other endpoint of the constraint. Since $u v x'' x'$ is empty, $u y$ crosses $v x''$. Since $s t$ is a visibility edge, $u y$ cannot cross it. Vertex $y$ cannot lie in $\overline{C}^u_1 \cap h'$, since this would imply that either $u y$ is an edge or there exists a vertex $z$ in the subcone of $y$ below $u y$ that contains $u$, which in combination with Lemma~\ref{lem:PathInTriangle} implies that there  exists a path between $y$ and $u$ that lies below $u y$. Since both alternatives contradict that $v$ minimizes $\angle x' u v$, $y$ cannot lie in $\overline{C}^u_1 \cap h'$. Hence, it remains to consider the case where $y$ lies in a subcone of $C^u_0$. Let $C^u_{0, j}$ be the subcone of $C^u_0$ to the right of $u y$. 
  
  If $y$ lies below $t$, $C^u_{0, j}$ contains a closest visible vertex whose angle with $u x'$ is less than $\angle x' u v$, contradicting that the routing algorithm routes to $v$. 

  If $y$ lies above $t$, let $z$ be the lowest vertex in the union of $C^u_{0, j}$ and $\overline{C}^u_1 \cap h'$. Since this region contains $t$, it is not empty and such a vertex $z$ exists. If $z \in C^u_{0, j}$, it is the closest vertex in $C^u_{0, j}$. If $z \in \overline{C}^u_1$, $u$ is the closest vertex to $z$. We note that in both cases $z$ is visible to $u$, since any constraint blocking it would have an endpoint below $z$. Hence, both cases result in an edge $u z$. However, since $\angle x' u z < \angle x' u v$, this contradicts that the routing algorithm routed to $v$. Thus, $u$ cannot be the endpoint of a constraint intersecting the interior of $u v x'' x'$.

  \textbf{Case (c):} If $v$ lies in a subcone of $C^u_0$ and $u v$ crosses $s t$, let $q$ be the intersection of $u v$ and $s t$. We need to show that the triangles $u q x'$ and $q x'' v$ are empty and do not contain any constraints (see Figure~\ref{fig:MaintainingInvariant}c). The proof that $u q x'$ is empty and does not contain any constraints is analogous to the previous case. 

  We prove that $q x'' v$ is empty by contradiction, so assume that $q x'' v$ is not empty. Since $q x''$ and $q v$ are visibility edges, we can apply Lemma~\ref{lem:ConvexChain} and we obtain a vertex $y$ in $q x'' v$ that is visible from $q$. If $y$ is visible from $u$, $v$ is not the closest vertex and edge $u v$ would not exist. If $y$ is not visible from $u$, we note that $u q$ is visible and apply Lemma~\ref{lem:ConvexChain} on triangle $u y q$. This gives us a vertex $z$ that is visible to $u$ and closer to $u$ than $v$, again contradicting the existence of edge $u v$. Hence, triangle $q x'' v$ is empty. 

  Finally, we show that $q x'' v$ does not contain any constraints. Since $q x''$ and $q v$ are visibility edges and $q x'' v$ is empty, any constraint intersecting the interior of $q x'' v$ must have $q$ as an endpoint. However, since $q$ is not a vertex, it cannot be the endpoint of a constraint. 
\end{proof}

Finally, we show that the path followed by the routing algorithm is 2-competitive, with respect to the Euclidean shortest path. 

\begin{theorem}
  \label{theo:PositiveCompetitive}
  Given two vertices $s$ and $t$ in the \graph such that $s$ and $t$ can see each other and $t$ lies in a positive subcone of $s$, there exists a 1-local routing algorithm that routes from  $s$ to $t$ and is 2-competitive with respect to the Euclidean distance. 
\end{theorem}
\begin{proof}[Proof.]
  We assume without loss of generality that $t \in C^s_0$. The routing algorithm will thus only take steps in $C^{v_i}_0$, $\overline{C}^{v_i}_1$, and $\overline{C}^{v_i}_2$, where $v_i$ is an arbitrary vertex along the routing path. Let $a$ and $b$ be the upper left and right corner of \canon{s}{t}. To bound the length of the routing path, we first bound the length of each edge. We consider three cases: (a) edges in subcones of $\overline{C}^{v_i}_1$ or $\overline{C}^{v_i}_2$, (b) edges in subcones of $C^{v_i}_0$ that do not cross $s t$, (c) edges in subcones of $C^{v_i}_0$ that cross $s t$. For ease of notation we use $v_0$ and $v_k$ to denote $s$ and $t$.

  \begin{figure}[ht]
    \begin{center}
      \includegraphics{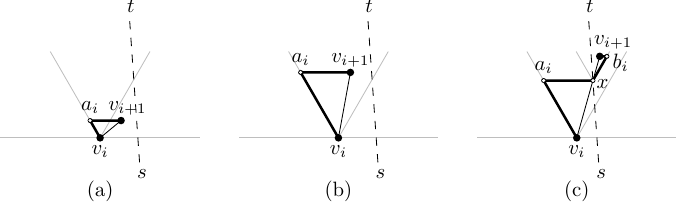}
    \end{center}
    \caption{Bounding the edge lengths: (a) an edge in a subcone of $\overline{C}^u_1$, (b) an edge in a subcone of $C^u_0$ that does not cross $s t$, and (c) an edge in a subcone of $C^u_0$ that crosses $s t$.}
    \label{fig:BoundingEdgeLength}
  \end{figure}

  \smallskip
  \textbf{Case (a):} If edge $v_i v_{i+1}$ lies in a subcone of $\overline{C}^{v_i}_1$, let $a_i$ be the upper corner of \canon{v_{i+1}}{v_i} (see Figure~\ref{fig:BoundingEdgeLength}a). By the triangle inequality, we have that $|v_i v_{i+1}| \leq |v_i a_i| + |a_i v_{i+1}|$. The case where $v_i v_{i+1}$ lies in $\overline{C}^{v_i}_2$ is analogous. 

  \textbf{Case (b):} If edge $v_i v_{i+1}$ lies in a subcone of $C^{v_i}_0$ and does not cross $s t$, let $a_i$ and $b_i$ be the upper left and right corner of \canon{v_i}{v_{i+1}} (see Figure~\ref{fig:BoundingEdgeLength}b). If $v_i$ lies to the left of $s t$, we use that $|v_i v_{i+1}| \leq |v_i a_i| + |a_i v_{i+1}|$. If $v_i$ lies to the right of $s t$, we use that $|v_i v_{i+1}| \leq |v_i b_i| + |b_i v_{i+1}|$. 

  \textbf{Case (c):} If edge $v_i v_{i+1}$ lies in a subcone of $C^{v_i}_0$ and crosses $s t$, we split it into two parts, one for each side of $s t$ (see Figure~\ref{fig:BoundingEdgeLength}c). Let $x$ be the intersection of $s t$ and $v_i v_{i+1}$. If $v_i$ lies to the left of $s t$, let $a_i$ be the upper left corner of \canon{v_i}{x} and let $b_i$ be the upper right corner of \canon{x}{v_{i+1}}. By the triangle inequality, we have that $|v_i v_{i+1}| \leq |v_i a_i| + |a_i x| + |x b_i| + |b_i v_{i+1}|$. If $v_i$ lies to the right of $s t$, let $a_i$ be the upper left corner of \canon{x}{v_{i+1}} and let $b_i$ be the upper right corner of \canon{v_i}{x}. By triangle inequality, we have that $|v_i v_{i+1}| \leq |v_i b_i| + |b_i x| + |x a_i| + |a_i v_{i+1}|$. 

  To bound the length of the full path, let $x$ and $x'$ be two consecutive points where the routing path crosses $s t$ and let $v_i v_{i+1}$ be the edge that crosses $s t$ at $x$ and let $v_{i'} v_{i'+1}$ be the edge that crosses $s t$ at $x'$. Let $a_x$ and $b_x$ be the upper left and right corner of \canon{x}{x'}. If the path between $x$ and $x'$ lies to the left of $s t$, this part of the path is bounded by: \[|x a_i| + \sum_{j=i}^{i'-1} |a_j v_{j+1}| + \sum_{j=i+1}^{i'} |v_j a_j| + |a_{i'} x'|.\] Since $x a_i$ and all $v_j a_j$ are parallel to $x a_x$ and all $a_x v_{j+1}$ are horizontal, we have that: \[|x a_i| + \sum_{j=i+1}^{i'} |v_j a_j| = |x a_x|.\] Similarly, since $a_{i'} x'$ and all $a_j v_{j+1}$ are parallel and have disjoint projections onto $a_x x'$, we have that: \[\sum_{j=i}^{i'-1} |a_j v_{j+1}| + |a_{i'} x'| = |a_x x'|.\] Thus, the length of a path to the left of $s t$ is at most: \[|x a_x| + |a_x x'|\] If the path between $x$ and $x'$ lies to the right of $s t$, this part of the path is bounded by (see Figure~\ref{fig:BoundingTotalLength}a): \[|x b_i| + \sum_{j=i}^{i'-1} |b_j v_{j+1}| + \sum_{j=i+1}^{i'} |v_j b_j| + |b_{i'} x'| = |x b_x| + |b_x x'|.\]

  \begin{figure}[ht]
    \begin{center}
      \includegraphics{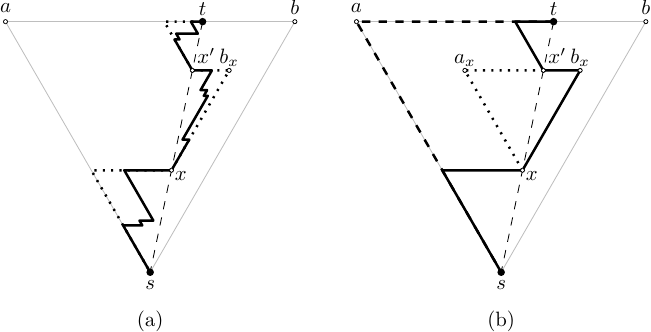}
    \end{center}
    \caption{Bounding the total length: (a) the bounds (solid lines) are unfolded (dotted lines) and (b) the unfolded bounds (solid lines) are flipped to the longer of the two sides (dotted lines) and unfolded again (dashed lines).}
    \label{fig:BoundingTotalLength}
  \end{figure}

  Next, we flip all unfolded bounds to the longer of the two sides $a t$ and $b t$: if $|a t| \geq |b t|$, we replace all bounds of the form $|x b_x| + |b_x x'|$ by $|x a_x| + |a_x x'|$ and if $|a t| < |b t|$, we replace all bounds of the form $|x a_x| + |a_x x'|$ by $|x b_x| + |b_x x'|$ (see Figure~\ref{fig:BoundingTotalLength}b). Note that this can only increase the length of the bounds. Finally, we sum these bounds and get that the sum is equal to $\max\{|s a| + |a t|, |s a| + |b t|\}$, which is at most $2 \cdot |s t|$. 
\end{proof}

\subsection[Routing on the Constrained $\Theta_6$-Graph]{Routing on the Constrained $\boldsymbol{\Theta_6}$-Graph}
To route on the constrained $\Theta_6$-graph, we split it into two constrained \graph{s}: the constrained \graph oriented as in Figure~\ref{fig:ConstrainedConesHalfGraph} and the constrained \graph where positive and negative cones are inverted. When we want to route from $s$ to $t$, we pick the constrained \graph in which $t$ lies in a positive subcone of $s$, referred to as $G^+$ in the remainder of this section, and apply the routing algorithm described in the previous section. Since this routing algorithm is 1-local and 2-competitive, we obtain a 1-local and 2-competitive routing algorithm for the constrained $\Theta_6$-graph, provided that we can determine locally, while routing, whether an edge is part of $G^+$. When at a vertex $u$, we consider the edges in order of increasing angle with the horizontal halfline through $u$ that intersects $s t$. 

\begin{lemma}
  \label{lem:EdgeInHalfThetaSix}
  While executing the positive routing algorithm for two visible vertices $s$ and $t$, we can determine locally at a vertex $u$ for any edge $u v$ in the constrained $\Theta_6$-graph whether it is part of $G^+$.  
\end{lemma}
\begin{proof}[Proof.]
  Suppose we color the edges of the constrained $\Theta_6$-graph red and blue such that red edges form $G^+$ and blue edges form the constrained \graph, where $t$ lies in a negative subcone of $s$. At a vertex $u$, we need to determine locally whether $u v$ is red. Since an edge can be part of both constrained \graph{s}, it can be red and blue at the same time. This makes it harder to determine whether an edge is red, since determining that it is blue does not imply that it is not red. 

  If $v$ lies in a positive subcone of $u$, we need to determine if it is the closest vertex in that subcone. Since by construction of the constrained \graph, $u$ is connected to the closest vertex in this subcone, it suffices to check whether this vertex is $v$. Note that if $u v$ is a constraint, $v$ lies in two subcones of $u$ and hence we need to check if it is the closest vertex in at least one of these subcones. 

  If $v$ lies in a negative subcone of $u$, we know that if it is not the closest visible vertex in that subcone, $u v$ is red. Hence, it remains to determine whether the edge to the closest vertex is red: If it is the closest visible vertex, it is blue, but it may be red as well if $u$ is also the closest visible vertex to $v$. Hence, we need to determine whether $u$ is the closest vertex in $C_{i, j}^v$, a subcone of $v$ that contains $u$. We consider two cases: (a) $u v$ is a constraint, (b) $u v$ is not a constraint. 

  \textbf{Case (a):} Since $u v$ is a constraint, we know that it cannot cross $s t$. Since we are considering $u v$, we also know that all edges that make a smaller angle with the horizontal halfline through $u$ that intersect $s t$ are not red. Hence, $u v$ is either part of the boundary of the routing path or the constraint is contained in the interior of the region bounded by the routing path and $s t$. However, by the invariant of Lemma~\ref{lem:RoutingStep}, the region bounded by the routing path and $s t$ does not contain any constraints in its interior. Thus, $u v$ is part of the boundary of the routing path and $u v$ is red. 

  \textbf{Case (b):} If $u v$ is not a constraint, let regions $A$ and $B$ be the intersection of $C_i^v$ and the two subcones of $u$ adjacent to $\overline{C}_i^u$ and let $C$ be the intersection of $C_{i, j}^v$ and the negative subcone of $u$ that contains $v$ (see Figure~\ref{fig:EdgeInHalfThetaSix}). We first note that since $u v$ lies in a negative subcone of $u$, the invariant of Lemma~\ref{lem:RoutingStep} implies that $B$ is empty. Furthermore, since $v$ is the closest visible vertex to $u$, $C$ does not contain any vertices that can see $u$ or $v$. 

  \begin{figure}[ht]
    \begin{center}
      \includegraphics[scale=0.8]{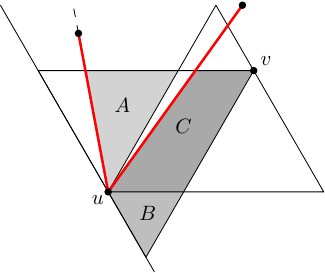}
    \end{center}
    \caption{The three regions $A$, $B$, and $C$ when determining whether an edge is part of the constrained \graph.}
    \label{fig:EdgeInHalfThetaSix}
  \end{figure}

  Since $C$ does not contain any vertices that can see $u$ or $v$, any constraint in $\overline{C}_i^u$ that has $u$ as an endpoint and lies above $u v$, ensures that $v$ cannot see $A$, i.e. it cannot block visibility of this region only partially. Hence, if such a constraint exists, $u$ is the closest visible vertex to $v$ in $C_{i, j}^v$, since neither $B$ nor $C$ contain any vertices visible to $v$. Therefore, $u v$ is red.  

  If $v$ can see $A$, we show that $u v$ is red if and only if the closest visible vertex in the subcone of $u$ that contains $A$ does not lie in $A$. We first show that 
if the closest visible vertex $x$ in the subcone of $u$ that contains $A$ lies in $A$, then $u v$ is not red. Since $A$ is visible to $v$, $u$ is not the endpoint of a constraint in $\overline{C}_i^u$ above $u v$. Hence, we have two visibility edges $u v$ and $u x$ and $u$ is not the endpoint of a constraint intersecting the interior of triangle $u x v$. Therefore, by Lemma~\ref{lem:ConvexChain}, we have a convex chain of visibility vertices between $x$ and $v$. Let $y$ be the vertex adjacent to $v$ along this chain. Since the polygon defined by $u x$, $u v$, and the convex chain is empty and does not contain any constraints, $y$ lies in $C_{i, j}^v$. Thus, $u$ is not the closest visible vertex in $C_{i, j}^v$ and $u v$ is not red. 

  Next, we show that if the closest visible vertex $x$ in the subcone of $u$ that contains $A$ does not lie in $A$, then $u v$ is red. We prove this by contradiction, so assume that $u v$ is not red. This implies that there exists a vertex $y \in C_{i, j}^v$ that is visible to $v$ and closer than $u$. Since $B$ is empty and $C$ does not contain any vertices that can see $v$, $y$ lies in $A$. Since $u v$ and $v y$ are visibility edges and $v$ is not the endpoint of a constraint intersecting the interior of triangle $u y v$, by Lemma~\ref{lem:ConvexChain} there exists a convex chain of visibility edges between $u$ and $y$. Furthermore, since $C$ does not contain any vertices that can see $u$, the vertex adjacent to $u$ along this chain lies in $A$. Since any vertex in $A$ is closer to $u$ than $x$, this leads to a contradiction, completing the proof. 
\end{proof}

\noindent \textbf{Routing Algorithm for the Constrained $\boldsymbol{\Theta_6}$-Graph} \\
Hence, to route on the constrained $\Theta_6$-graph, we apply the positive routing algorithm on $G^+$, while determining which edges are part of this constrained \graph. The latter can be determined as follows: If $v$ lies in a positive subcone, we need to check whether it is the closest vertex in that subcone. If $v$ lies in a negative subcone and it is not the closest vertex, it is part of the constrained \graph. Finally, if $v$ is the closest vertex in a negative subcone, it is part of the constrained \graph if it is a constraint or the intersection of the cone of $v$ that contains $u$ and the subcone of $C_{i-1}^u$ adjacent to $\overline{C}_i^u$ is empty.

\subsection[Negative Routing on the Constrained \graph]{Negative Routing on the Constrained Half-$\boldsymbol{\Theta_6}$-Graph}
We note that the routing algorithm provided in the previous section does not suffice to also route on the constrained \graph, since it assumes that the destination lies in a positive subcone of the source. Therefore, in this section, we provide an $O(1)$-memory 1-local routing algorithm for the case where the destination $t$ lies in a negative subcone of the source $s$. 

For ease of exposition, we assume that $t$ lies in a subcone of $\overline{C}^s_0$. The $O(1)$-memory 1-local routing algorithm finds a path from $s$ to $t$ of length at most $2 \cdot |s t|$ and travels a total distance of at most $18 \cdot |s t|$ to do so. We note that negative routing is harder than positive routing, since there need not be an edge to a vertex in the cone of $s$ that contains $t$. This phenomenon also caused the separation between spanning ratio and routing ratio in the unconstrained setting~\cite{BFRV2015RoutingJournal}. 

The remainder of this section is structured as follows: First, we identify a set of conditions that edges need to meet in order to be considered by the routing algorithm. Unfortunately, one of these conditions cannot be checked 1-locally. Therefore, we replace it with a set of conditions that exclude edges that are guaranteed not to satisfy the original condition and can be checked 1-locally. 

We proceed to describe the edges considered by the negative routing algorithm. Given a vertex $v$ and all neighbors of $v$ whose projection along the bisector of $C^t_0$ is closer to $t$ than the projection of $v$, we number the neighbors $u_0, ..., u_k$ of $v$ in counterclockwise order, starting from the horizontal half-line to the left of $v$ (see Figure~\ref{fig:LastRegion}). We create $k+2$ regions around $v$:
\begin{itemize}
  \item We create $k$ triangular regions $v u_i u_{i+1}$ for $0 \leq i \leq k-1$.
  \item We create one unbounded region using edge $v u_0$ and the two horizontal half-lines starting at $v$ and $u_0$ directed towards the left. 
  \item We create one unbounded region using edge $v u_k$ and the two horizontal half-lines starting at $v$ and $u_k$ directed towards the right. 
\end{itemize}

\begin{figure}[ht]
  \begin{center}
    \includegraphics{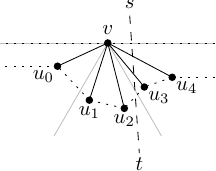}
  \end{center}
  \caption{Triangle $v u_2 u_3$ is the last region of $v$ intersected by $s t$.}
  \label{fig:LastRegion}
\end{figure}

The \emph{last region} of $v$ intersected by $s t$ is defined as the last of these regions that is encountered when following $s t$ from $s$ to $t$. In Figure~\ref{fig:LastRegion}, the region defined by $v$, $u_2$, and $u_3$ is the last region of $v$ intersected by $s t$.

We consider an edge $u v$ for our routing algorithm when it satisfies the following three conditions: 
\begin{enumerate}
  \item Vertices $u$ and $v$ lie inside or on the boundary of \canon{t}{s}. 
  \item Edge $u v$ is part of the last region of $v$ that is intersected by $s t$. 
  \item Edge $u v$ is the edge that the positive routing algorithm picks at $u$ when routing from $t$ to $s$. Note that for this condition, we do not require that $u$ is part of the positive routing path, but only that should the positive routing path reach $u$, edge $u v$ is the edge it would select for its next step. 
\end{enumerate}
Given $s$ and $t$, the first two requirements can be checked using only the location of $s$ and $t$ and 1-local information, i.e. the neighbors of the current vertex. The last requirement, on the other hand, may need 2-local information as it involves the neighbors of the neighbors of $v$. Hence, instead of using this last requirement, we ignore the edges that can never satisfy it and show that we can route competitively and 1-locally on the graph $G$ formed by the edges that meet the first two requirements. 

Since $t$ lies in a subcone of $\overline{C}^s_0$, the edges that define the last intersected region of a vertex $v$ can lie in three cones: $C^v_1$, $\overline{C}^v_0$, and $C^v_2$. Since edges in $C^v_1$ and $C^v_2$ of the negative routing algorithm correspond to edges in $\overline{C}^u_1$ and $\overline{C}^u_2$ of the positive routing algorithm (applied from t to s), the positive routing algorithm never follows these edges if they intersect $s t$. Hence, these edges need not be considered by the negative routing algorithm (see Figure~\ref{fig:IgnoredEdges}a). 

\begin{figure}[ht]
  \begin{center}
    \includegraphics{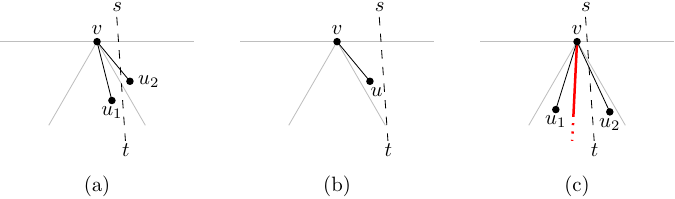}
  \end{center}
  \caption{The edges ignored by the negative routing algorithm: (a) edge $u_2 v$ is ignored since it intersects $s t$, (b) edge $u v$ is ignored since $C^v_2$ is intersected by $s t$, (c) edge $u_1 v$ is ignored since it lies in a subcone that is not intersected by $s t$ and $u_1 v u_2$ is intersected by a constraint that has $v$ as an endpoint.}
  \label{fig:IgnoredEdges}
\end{figure}

We also do not need to consider edges in $C^v_1$ and $C^v_2$ when that cone is intersected by $s t$ (see Figure~\ref{fig:IgnoredEdges}b): Assume $C^v_1$ is intersected by $s t$. Since we are considering edge $u v$, it cannot cross $s t$. Hence, $s t$ intersects cone $C^u_1$, but more importantly $s t$ intersects $\overline{C}^u_2$. Hence, if the positive routing algorithm reaches $u$, it continues by following an edge in $\overline{C}^u_2$ or $C^u_0$. Since $C^v_1$ corresponds to $\overline{C}^u_1$, no edge in this cone is followed by the positive routing algorithm, and we can ignore it. 

Finally, we ignore edges that lie in a subcone that is not intersected by $s t$ when $v$ is the endpoint of a constraint that intersects the interior of the last region of $v$ that is intersected by $s t$ (see Figure~\ref{fig:IgnoredEdges}c): If $v$ is the endpoint of a constraint that intersects the interior of the last region of $v$ that is intersected by $s t$, we do not consider the edge that is not intersected by $s t$. We can ignore this edge, since by the invariant, the region between the routing path and $s t$ does not contain any constraints. 

Since these conditions can be checked using only $s$, $t$, $v$, the neighbors of $v$, and the constraints incident to $v$, we can determine 1-locally whether to consider an edge. Hence, the graph $G$ on which we route is the graph formed by all edges $u v$ such that: 
\begin{enumerate}
  \item Vertices $u$ and $v$ lie inside or on the boundary of \canon{t}{s}. 
  \item Edge $u v$ is part of the last region of $v$ that is intersected by $s t$. 
  \item Edge $u v$ does not meet any of the following three conditions: 
	\begin{enumerate}
	  \item \label{cond:3a} Edge $u v$ lies in $C^v_1$ or $C^v_2$ and crosses $s t$. 
	  \item \label{cond:3b} Edge $u v$ lies in $C^v_1$ or $C^v_2$ and this cone is intersected by $s t$.
	  \item Edge $u v$ lies in a subcone that is not intersected by $s t$ and $v$ is the endpoint of a constraint that intersects the interior of the last region of $v$ that is intersected by $s t$. 
	\end{enumerate}
\end{enumerate}

Note that every edge $u v$ that lies in $C^v_1$ or $C^v_2$ and crosses $s t$, the cone that contains $u v$ is intersected by $s t$. Hence, condition~\ref{cond:3a} can be ignored as it is included in condition~\ref{cond:3b}.  

In the remainder of this section, for ease of exposition, we consider each edge of $G$ to be oriented upward: Let $u'$ and $v'$ be the projections of $u$ and $v$ along the bisector of $C^t_0$. Edge $u v$ is oriented from $u$ to $v$ if and only if $|t u'| \leq |t v'|$. Note that this does not imply that $u$ lies in a negative cone of $v$. We proceed to prove that every vertex with two incoming edges is part of the positive routing path when routing from $t$ to $s$. 

\begin{lemma}
  Every vertex with in-degree 2 in $G$ that is reached by the negative routing algorithm is part of the positive routing path from $t$ to $s$. 
\end{lemma}
\begin{proof}[Proof.]
  Let $v$ be a vertex of in-degree 2 that is reached by the negative routing algorithm. Let $u$ and $w$ be the other endpoints of these edges to $v$, such that the projection of $u$ along the bisector of \canon{t}{s} is closer to $t$ than the projection of $w$ (see Figure~\ref{fig:InDegree2}). Since both $u v$ and $w v$ are part of the last intersected region of $v$, vertices $u$ and $w$ must lie on opposite sides of $s t$. This implies that the positive routing algorithm reaches at least one of them when routing from $t$ to $s$, since by the invariant the region between the routing path and $s t$ is empty. Thus it suffices to show that from both $u$ and $w$ the positive routing algorithm eventually reaches $v$. 

  \begin{figure}[ht]
    \begin{center}
      \includegraphics{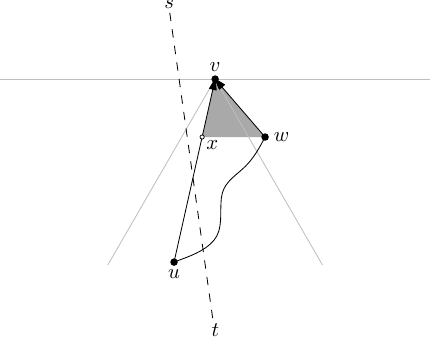}
    \end{center}
    \caption{Vertex $v$ has in-degree 2.}
    \label{fig:InDegree2}
  \end{figure}

  If the positive routing algorithm reaches $w$, we show that it would follow the edge to $v$. Let $x$ be the intersection of $u v$ and the horizontal line through $w$ (see Figure~\ref{fig:InDegree2}). First, we show that triangle $v w x$ is empty. If $w$ lies in a subcone of $C^v_1$ or $C^v_2$, $u$ lies in a subcone of $\overline{C}^v_0$, since otherwise one of the two edges would cross $s t$ and be ignored. Since $v w$ and $v x$ are visibility edges and $v$ is not the endpoint of a constraint intersecting the interior of $v w x$, it follows from Lemma~\ref{lem:ConvexChain} that if $v w x$ is not empty, there exists a convex chain of visibility edges between $w$ and $x$ and the region bounded by this chain, $v w$, and $v x$ is empty. Let $y$ be the topmost vertex along this convex chain and note that $y$ is visible to $v$. If $y$ lies in the same cone of $v$ as $w$, it also lies in the same subcone of $v$ as $w$, since $v$ is not the endpoint of a constraint intersecting the interior of $v w x$. However, this implies that $w$ is not the closest visible vertex to $v$ in this subcone, contradicting that $v w$ is an edge. If $y$ lies in $\overline{C}^v_0$, $y$ has an edge in its subcone that contains $v$, since $v$ is a visible vertex in that subcone. This edge cannot cross $v w$ and $v u$, since the constrained \graph is plane, and it cannot be connected to a vertex in the region bounded by the convex chain, $v w$, and $v x$, since it is empty. Finally, since $y$ is the topmost vertex along the convex chain, the edge cannot connect $y$ to another vertex of the convex chain. Hence, $y$ would have an edge to $v$, contradicting that $v u$ and $v w$ are consecutive edges around $v$. We conclude that triangle $v w x$ is empty. 

  Using an analogous argument, it can be shown that if $u$ lies in a subcone of $C^v_1$ or $C^v_2$, $w$ lies in $\overline{C}^v_0$ and the existence of a vertex in $v w x$ would contradict that $u v$ is an edge or that $u$ and $w$ are consecutive edges around $v$. If both $u$ and $w$ lie in a subcone of $\overline{C}^v_0$, the argument reduces to the case where $y$ lies in $\overline{C}^v_0$, again contradicting that $u$ and $w$ are consecutive edges around $v$. Hence, since $v w x$ is empty, the positive routing algorithm routes to $v$ when it reaches $w$, since it minimizes angle $\angle x w v$. 

  Next, we look at the case where the positive routing path reaches $u$. If it follows edge $u v$, we are done. If it does not follow edge $u v$, let $z$ be the other endpoint of the edge the positive routing algorithm follows at $u$. By construction of the positive routing path, we know that the projection of $z$ on the bisector of $C^t_0$ lies further from $t$ than the projection of $u$. Since the constrained \graph is plane, the path from $z$ to $s$ cannot cross $u v$ or $w v$, and since the positive routing path is monotone with respect to the bisector of $C^t_0$, it cannot go down and around or through $u$. Furthermore, since the region enclosed by the positive routing path and $s t$ is empty, the path also cannot go around $w$ without passing through $w$. Finally, since $u v$ and $w v$ are consecutive edges around $v$, the path from $z$ to $s$ cannot reach $v$ by arriving from an edge between $u v$ and $w v$. Hence, $w$ must lie on the path from $z$ to $s$. Thus, since we previously showed that when the positive routing algorithm reaches $w$, it routes to $v$, vertex $v$ is also reached when the positive routing path reaches $u$. 
\end{proof}

\noindent \textbf{Negative Routing Algorithm for the Constrained Half-$\boldsymbol{\Theta_6}$-Graph} \\
Routing from $s$ to $t$ now comes down to searching for a vertex that has in-degree 2 on one of the two paths leaving $s$. When such a vertex $v$ is found, we need to find the next vertex that has in-degree 2 on one of the two paths leaving $v$. This process is repeated until we reach $t$. A single instance of this problem, i.e. finding the next vertex has in-degree 2 from another vertex can be viewed as searching for a specific point on a line. This problem has been studied extensively and a search strategy that is 9-competitive was presented by Baeza-Yates~\etal~\cite{BCR93}: We start by following the shorter of the two edges connected to $s$ and call this distance 1. If we reached a vertex with in-degree 2, we are done. Otherwise, we go back to $s$ and follow the other path up to distance 2 from $s$. Again, if we reached a vertex with in-degree 2, we are done. Otherwise, we go back to $s$ and follow the first path up to distance 4 from $s$. This process of backtracking and doubling the allowed travel distance is repeated until a vertex with in-degree 2 is reached. Since this strategy needs to keep track of the distance traveled, it uses $O(1)$-memory. Hence, we apply this search strategy and perform the following actions when we reach an unvisited vertex $v$:
\begin{itemize}
  \item If $v$ has in-degree 2, $v$ is part of the positive routing path and we restart the searching strategy from $v$. 
  \item If $v$ has in-degree 1, we proceed to its neighbor $u$ if we have enough budget left to traverse the edge. At $u$ we check whether the positive routing algorithm would follow edge $u v$. If this is not the case, we know that $v$ was a dead end and the path on the opposite side of $s t$ is part of the positive routing path. Hence, we backtrack and follow the path on the opposite side of $s t$ to the last visited vertex on that side. 
  \item If $v$ has in-degree 0, it is a dead end and we backtrack like in the previous case. 
\end{itemize}

We conclude this section by showing that the above $O(1)$-memory 1-local routing algorithm has a routing ratio of at most 9 times the length of the positive routing path, which implies an 18-competitive 1-local routing algorithm for negative routing in the constrained \graph. 

\begin{theorem}
  There exists an $O(1)$-memory 1-local 18-competitive routing algorithm for negative routing in the constrained \graph between vertices that can see each other. 
\end{theorem}
\begin{proof}[Proof.]
  Let $p$ be the last vertex where the search strategy was restarted --- initially $p$ is $s$. We prove the theorem by showing that when we restart the search strategy at vertex $q$, we traveled at most 9 times the distance along the positive routing path between $p$ and $q$. If we restart the search strategy because we reached a vertex of in-degree 2, this follows directly from the fact that the search strategy is 9-competitive, i.e. we found the vertex we are looking for and we spent at most 9 times the distance along the positive routing path between $p$ and $q$. 

  If we reach a vertex $v$ with in-degree 0 or we traverse an edge $v u$ and the positive routing algorithm would not have routed from $u$ to $v$, we backtrack to $p$ and traverse the path on the opposite side of $s t$. We follow this path until we reach $w$, the last vertex traversed on this side of $s t$. Unfortunately, $w$ is too close to $p$ to prove that the total length traveled is at most 9 times the distance along the positive routing path between $p$ and $w$. However, $w$ must have in-degree 1: Since $w$ is part of the positive routing path, it cannot have in-degree 0, and since we did not restart the search strategy when we reached $w$ the previous time, it cannot have in-degree 2. Hence, it has in-degree 1 and it follows that the vertex $q$ to which $w$ is connected is also part of the positive routing path. Since the distance along the positive routing path between $p$ and $v$ is at most 2 times the distance along the positive routing path between $p$ and $q$, an argument analogous to the one used by Baeza-Yates~\etal~\cite{BCR93} shows that we traversed at most 9 times the distance along the positive routing path between $p$ and $q$ to reach $q$. 
\end{proof}

\subsection{Lower Bound on the Negative Routing Algorithm}
In this section we show that the negative routing algorithm described in the previous section cannot be guaranteed to reach $t$ while traveling less than $2 \sqrt{39} \cdot |s t| \approx 12.48 \cdot |s t|$. This situation is shown in Figure~\ref{fig:SimplifiedWorst}: We place a vertex $r_1$ almost horizontally to the right of $s$ at distance 1, followed by a vertex $l_1$ almost horizontally to the left of $s$ at distance 2, followed by a vertex $r_2$ almost horizontally to the right of $s$ at distance 4. Once we reach the corners of \canon{t}{s} at $l_2$ and $r_3$, we proceed down along the boundary of \canon{t}{s} and place vertices $l_3$ and $r_4$ such that the distance between $s$ and $l_3$ via $l_2$ is 8 and the distance between $s$ and $r_4$ via $r_3$ is 16. Finally, we place vertices $l_4$ and $r_5$ arbitrarily close to $t$. The positive routing path from $t$ to $s$ would route to $r_5$, $r_4$, $r_3$, $r_2$, $r_1$, and finally $s$. 

\begin{figure}[ht]
  \begin{center}
    \includegraphics{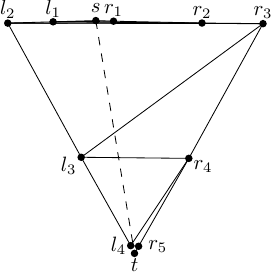}
  \end{center}
  \caption{The situation where the negative routing algorithm uses $2 \sqrt{39} \cdot |s t|$ to reach $t$.}
  \label{fig:SimplifiedWorst}
\end{figure}

The negative routing algorithm on the other hand would try both sides, going back through $s$ each time it switches sides: go to $r_1$, go to $l_1$, go to $r_2$, go to $l_3$ (via $l_2$), go to $r_5$ (via $r_3$), go to $l_4$ (via $l_2$), and finally go to $t$ (via $s$ and $r_5$). We can pick the edge lengths between the vertices in such a way that each time the next vertex along one of the two sides is reached (other than $l_4$), the negative routing algorithm runs out of budget and needs to backtrack to try the other side. The total length traveled this way is the sum of: 
\begin{itemize}
  \item $2 \cdot \delta(s, r_5)$, for going back and forth from $s$ until the step before $r_5$ is reached for the first time, 
  \item $2 \cdot \delta(s, r_5)$, for going to $r_5$ and back to $s$ when the negative routing algorithm almost reaches $t$, 
  \item $2 \cdot \delta(s, l_4)$, for going down the wrong path (and back up) after reaching $r_5$, 
  \item $\delta(s, t)$, for finally reaching $t$, 
\end{itemize}
where $\delta(x, y)$ is the distance along the negative routing path between $x$ and $y$. Since $r_5$ can be arbitrarily close to $t$, this sums up to $5 \cdot \delta(s, t) + 2 \cdot \delta(s, l_4)$. 

Let $\alpha$ be the angle between the bisector of \canon{t}{s} and $t s$. Using the law of sines, we can express $\delta(s, t)$ and $\delta(s, l_4)$ as follows: 
\begin{eqnarray*}
  \delta(s, t) &=& |s r_3| + |r_3 t|\\
  &=& \left( \frac{\sin\left( \frac{\pi}{6} + \alpha \right)}{\sin \left( \frac{\pi}{3} \right)} + \frac{\sin\left( \frac{\pi}{2} - \alpha \right)}{\sin \left( \frac{\pi}{3} \right)} \right) \cdot |s t| \\
  &=& \left( \sqrt{3} \cdot \cos \alpha + \sin \alpha \right) \cdot |s t| \\
  && \\
  \delta(s, l_4) &=& |s l_2| + |l_2 l_4|\\
  &=& \left( \frac{\sin\left( \frac{\pi}{6} - \alpha \right)}{\sin \left( \frac{\pi}{3} \right)} + \frac{\sin\left( \frac{\pi}{2} - \alpha \right)}{\sin \left( \frac{\pi}{3} \right)} \right) \cdot |s t| \\
  &=& \left( \sqrt{3} \cdot \cos \alpha - \sin \alpha \right) \cdot |s t| 
\end{eqnarray*}
Thus, the total distance traveled by the negative routing algorithm becomes: 
\begin{eqnarray*}
  & & 5 \cdot \delta(s, t) + 2 \cdot \delta(s, l_4) \\
  &=& 5 \cdot \left( \sqrt{3} \cdot \cos \alpha + \sin \alpha \right) \cdot |s t| + 2 \cdot \left( \sqrt{3} \cdot \cos \alpha - \sin \alpha \right) \cdot |s t| \\
  &=& \left( 7 \sqrt{3} \cdot \cos \alpha + 3 \sin \alpha \right) \cdot |s t|
\end{eqnarray*}
When maximizing this function over $\alpha$, with $0 \leq \alpha \leq \pi/6$, we find the maximum at $\alpha \approx 0.2425$, where the function has value $2 \sqrt{39} \cdot |s t| \approx 12.48 \cdot |s t|$.

\section{Conclusion}
We showed that no deterministic 1-local routing algorithm is $o(\sqrt{n})$-competitive on all pairs of vertices of the constrained $\Theta_6$-graph, regardless of the amount of memory it is allowed to use. Following this negative result, we showed how to route between any two \emph{visible} vertices of the constrained $\Theta_6$-graph using only 1-local information by routing on one of the two constrained \graph{s}. This routing algorithm guarantees that the returned path has length at most 2 times the Euclidean distance between the source and destination. Additionally, we provided a 1-local 18-competitive routing algorithm for visible vertices in the constrained half-$\Theta_6$-graph. To the best of our knowledge, this is the first 1-local routing algorithm in the constrained setting with guarantees on the length of the returned path. 

There remain a number of open problems in the area of local competitive routing in the constrained setting. For example, though we showed that no deterministic 1-local routing algorithm is $o(\sqrt{n})$-competitive on all pairs of vertices of the $\Theta_6$-graph, it would still be interesting to construct a routing algorithm that reaches any vertex. 

Furthermore, we showed how to route on a specific constrained $\Theta$-graph. It would be very nice if there exists a local routing algorithm that is competitive on all constrained $\Theta$-graphs. In the unconstrained setting, the $\Theta$-routing algorithm (which repeatedly follows the edge to the closest vertex in the cone that contains the destination) is such an algorithm, provided that at least 7 cones are being used. In the constrained setting, however, this particular algorithm need not reach the destination, since even if the source can see the destination, this does not necessarily hold for every vertex along the path. Because of this, there need not be any edge in the cone that contains the destination, meaning that this $\Theta$-routing algorithm can get stuck. 

Finally, constrained $\Theta$-graphs are not the only graphs that are known to be spanners in the constrained setting. For example, constrained Yao-graphs and constrained (generalized) Delaunay graphs have also been shown to be spanners~\cite{R2014,R2014Thesis}. As was the case for constrained $\Theta$-graphs prior to our work, no routing algorithms are known to exist for those graphs.

\bibliographystyle{plain}
\bibliography{route}

\end{document}